\documentclass[journal,10pt]{IEEEtran}
\IEEEoverridecommandlockouts
\usepackage{cite}
\usepackage{dblfloatfix}
\usepackage{subcaption}
\usepackage{overpic}
\usepackage{amsmath,amssymb,amsfonts}
\usepackage{graphicx}
\usepackage{textcomp}
\usepackage{xcolor}
\usepackage{float}
\usepackage{amsthm}
\usepackage{graphicx}
\usepackage{epstopdf}
\usepackage{amsmath,bm}
\usepackage{amsfonts}
\usepackage{amssymb}
\usepackage{color}
\usepackage{multirow}
\usepackage{multicol}
\usepackage{soul,xcolor}
\usepackage{algorithm}
\usepackage{algpseudocode}

\usepackage{comment}

\usepackage[subtle]{savetrees}

\theoremstyle{plain}

\newtheorem{lemma}{Lemma}

\newtheorem{corollary}{Corollary}

\newcommand{\vect}[1]{\mathbf{#1}}

\def\Htran{\mbox{\tiny $\mathrm{H}$}}

\def\CN{\mathcal{N}_{\mathbb{C}}}

\begin{document}

\title{Stream-Adaptive Quantization and Power Allocation in Fronthaul-Constrained MIMO Systems}

\author{\"Ozlem Tu\u{g}fe Demir and Emil Bj{\"o}rnson 
 
\thanks{  \"O. T. Demir is with the Department of Electrical and Electronics Engineering, Bilkent University, Ankara, T\"urkiye (ozlemtugfedemir@bilkent.edu.tr).  E.~Bj\"ornson is with the Department of Communication Sytems, KTH Royal Institute of Technology, 10044 Stockholm, Sweden (emilbjo@kth.se).  \newline \indent
This work was carried out within the scope of the project 122C149 – Intelligent End-to-End Design of Energy-Efficient and Hardware Impairments-Aware Cell-Free Massive MIMO for Beyond 5G. \"O. T. Demir was supported by the 2232-B International Fellowship for Early Stage Researchers Programme funded by the Scientific and Technological Research Council of Türkiye (TÜBİTAK).  The work by E. Björnson was supported by a fellowship from the Knut and Alice Wallenberg Foundation.
}
}

\maketitle
\begin{abstract} 
Many wireless systems divide the baseband processing between two locations, interconnected by a fronthaul.
This paper examines the impact of fronthaul quantization on multiple-input multiple-output (MIMO) systems. Starting from a Bussgang-based analysis of quantized single-input single-output (SISO) channels, we extend the framework to MIMO and derive a capacity lower bound under fronthaul quantization, where the receive combining is performed before the quantization. To maximize the sum rate, we propose a joint bit and power allocation (JBP-Alloc) scheme that efficiently distributes fronthaul bits and transmit power across active data streams. Asymptotic analysis shows that uniform bit allocation becomes optimal at high SNR. Numerical results confirm that JBP-Alloc outperforms uniform allocation and quantization-unaware water-filling, and achieves the same performance as Greedy bit allocation but with substantially lower computational complexity.
\end{abstract}
\begin{IEEEkeywords}Fronthaul limitation, quantization, energy efficiency, MIMO capacity, multi-stream.
\end{IEEEkeywords}

\section{Introduction}

Multiple-input multiple-output (MIMO) systems enable both spatial multiplexing and beamforming gains, thereby achieving much higher data rates compared to the single-input single-output (SISO) counterpart. The capacity of a point-to-point MIMO channel is achieved by precoding based on the singular-value decomposition (SVD), water-filling power allocation, and separate decoding of the multiplexed data streams \cite{telatar1999capacity}. 
In modern Open RAN systems, the MIMO processing is often divided between different places: receive combining is done in the radio unit, and decoding is done in the baseband unit. The capacity-limited fronthaul connection between these units imposes additional quantization errors \cite{khorsandmanesh2023optimized}, which are not captured by the classical MIMO theory.

Quantization of the received uplink signal has been extensively studied in the context of analog-to-digital converters (ADCs) \cite{mo2015capacity,studer2016quantized,nossek2006capacity}, where quantization is applied directly to the received signal.
However, contemporary base stations typically use 12-16 ADC bits to enable advanced digital filtering, in which case the quantization errors are negligible compared to the thermal noise.
However, when the received signals are sent over a limited-capacity fronthaul link, it is desirable to lower the bit resolution, at the cost of creating quantization errors.
This problem has new characteristics: we can preprocess the received signals before quantization, and also use varying bit resolution depending on the channel conditions.
In a MIMO system, we can multiply the received signal by the left singular matrix of the channel to separate the data streams, which can then be quantized independently before being transmitted over the fronthaul. 
This raises important research questions: how should the total number of quantization bits and the transmit power be allocated across the data streams, and how should the number of active streams be determined? 
We will show that classical water-filling is suboptimal in this setting and that significant improvements can be made. 

The main contributions of this paper are:
\begin{itemize}
    \item We derive a lower bound on the achievable rate of quantized MIMO channels with fronthaul quantization.
    \item We formulate the joint optimization problem of transmit power and quantization bit allocation across the MIMO data streams, subject to total fronthaul and power constraints, and propose an efficient algorithm to solve it.
    \item Through asymptotic analysis, we prove that uniform bit allocation becomes optimal in the high-SNR regime.
\end{itemize}

\textit{Outline:} Section~\ref{section:SISO} analyzes the SISO case and derives a mutual information bound. Section~\ref{section:MIMO} extends the analysis to MIMO channels. Section~\ref{section:allocation} formulates and solves the joint bit and power allocation problem. Section~\ref{eq:section-asymptotic} analyzed the asymptotic behavior. Section~\ref{section:results} presents numerical results, and Section~\ref{section:conclusions} concludes the paper.

{\textit{Reproducible Research:} All the simulation results can be reproduced using the
Matlab code available at: https://github.com/ozlemtugfedemir/stream-adaptive-quantization

\section{SISO Analysis and Capacity Lower Bound  }
\label{section:SISO}

We will start our analysis by considering a SISO system and derive a lower bound on the mutual information between the input signal and the quantized noisy output signal. We denote the received signal before quantization as
\begin{align}
    y = hx+n,
\end{align}
where $x\sim\CN(0,P)$ is the desired data signal with power $P$, $h\in \mathbb{C}$ is the deterministic channel coefficient, and $n\sim\CN(0,\sigma^2)$ is the additive independent thermal noise.

We assume that the received signal is quantized before being sent to the baseband unit via the fronthaul. Since the quantization is not done by an ADC but software, we can consider an optimal non-uniform quantizer that takes the Gaussian input and minimizes the distortion noise, which is known as the Lloyd-Max quantizer. This quantizer, which is designed for the input $y$, is denoted by $Q(\cdot)$ and the quantized signal becomes
\begin{align} \label{eq:Qy}
    z = Q(y).
\end{align}
Each quantization interval is represented by its
mean value \cite{demir2020bussgang}: $\mathbb{E}\{y|Q(y)\}=Q(y)$.\footnote{Although we assume a Lloyd–Max quantizer, the subsequent analysis relies only on the conditional mean property implied by Lloyd–Max optimality, rather than on the full optimality of the quantizer design. In particular, this property holds whenever the quantization decoder mapping is optimal, regardless of whether the encoder mapping is optimal \cite{Fletcher2007a,gersho2012vector}. Moreover, the same conditional mean condition holds approximately for uniform scalar quantizers operating in the high-resolution regime \cite{Fletcher2007a}.} The distortion factor $\beta$ is defined as

\begin{align}
    \beta = \frac{\mathbb{E}\{|z-y|^2\}}{C_y},
\end{align}
where $C_y=\mathbb{E}\{|y|^2\}=P|h|^2+\sigma^2$.
Next, we find a lower bound on the mutual information between $x$ and $z$. 

 We can express the output of the quantizer function in \eqref{eq:Qy} using the Bussgang decomposition \cite{demir2020bussgang} as
\begin{align}
    z = By+\eta = B(hx+n)+\eta,
\end{align}
where $B$ is the Bussgang gain and $\eta$ is the distortion noise, which is uncorrelated with $y$ by construction. The Bussgang gain $B$ is computed as \cite{demir2020bussgang}
\begin{align}
    B= \frac{\mathbb{E}\{zy^*\}}{C_y}.
\end{align}
Then, one can easily show that
\begin{align}
    \mathbb{E}\{\eta y^*\} &= \mathbb{E}\{(z-By)y^*\} = \mathbb{E}\{zy^*\}-BC_y \nonumber\\
    &= \mathbb{E}\{zy^*\}-\frac{\mathbb{E}\{zy^*\}}{C_y}C_y = 0.
\end{align}
The next lemma shows that $\eta$ is also uncorrelated with $x$.

\begin{lemma} \label{lemma:SISO}
The distortion noise $\eta$, which is uncorrelated with $y$, is also uncorrelated with $x$, i.e., $\mathbb{E}\{\eta x^*\}=0$.
\begin{proof}
    We follow the scalar version of the proof provided for channel estimation in \cite[App.~A]{li2017channel}, adapting it to the case of data detection. We start by writing
    \begin{align}
        \mathbb{E}\{\eta x^*\}& = \mathbb{E}\{(Q(y)-By)x^*\} \nonumber\\
        & = \mathbb{E}\{\mathbb{E}\{(Q(y)-By)x^*| y\}\} \nonumber \\
        &=\mathbb{E}\{(Q(y)-By)\mathbb{E}\{x^*| y\}\} \label{eq:iterated-expectation},
    \end{align}
where we have used the law of iterated expectation, and that once $y$ is given, $(Q(y)-By)$ is deterministic. Since $y=hx+n$ is a Gaussian signal corrupted by Gaussian noise, $\mathbb{E}\{x^*| y\}$ in \eqref{eq:iterated-expectation}, which is the minimum mean-squared error (MMSE) estimate of $x^*$ given $y$, is also the linear MMSE, given as
\begin{align}
    \mathbb{E}\{x^*|y\} = \frac{Ph}{C_y}y^*. \label{eq:MMSE}
\end{align}
Inserting the MMSE estimate in \eqref{eq:MMSE} into \eqref{eq:iterated-expectation}, we prove that
\begin{align}
        \mathbb{E}\{\eta x^*\} = \mathbb{E}\left\{\eta \frac{Ph}{C_y} y^*\right\} = \frac{Ph}{C_y}\mathbb{E}\{\eta y^*\} = 0,
        \end{align}
        where we use the uncorrelatedness of $\eta$ of $y$.
\end{proof}
\end{lemma}

Since $\eta$ is uncorrelated with $x$, due to Lemma~\ref{lemma:SISO}, we can express the quantized signal as
\begin{align}
    z= Bhx + Bn+\eta
\end{align}
where $Bn+\eta$ is the effective uncorrelated noise. Then, using the worst-case uncorrelated additive noise theorem \cite{Hassibi2003a}, we can obtain a lower bound on the mutual information as
\begin{align}
   I(x;z)\geq \log_2\left(1+\frac{|B|^2P|h|^2}{C_z-|B|^2P|h|^2}\right) ,\label{eq:information-lower-bound}
\end{align}
where we use the notation $C_z= \mathbb{E}\{|z|^2\}$. 

 As shown in \cite{demir2020bussgang}, the variance of $z$ is
\begin{align}
   & C_z = (1-\beta)C_y= (1-\beta)(P|h|^2+\sigma^2). \label{eq:Cz}
\end{align}
The Bussgang gain $B$ can be shown to be equal to $(1-\beta)$ by using the property $\mathbb{E}\{y|Q(y)\}=Q(y)$  of the quantizer. Then the mutual information lower bound is obtained as a function of $P$, $\beta$, $h$, and $\sigma^2$ as
\begin{align}
   I(x;z)&\geq \log_2\left(1+\frac{(1-\beta)^2P|h|^2}{(1-\beta)(P|h|^2+\sigma^2)-(1-\beta)^2P|h|^2}\right) \nonumber\\
   &= \log_2\left(1+\frac{(1-\beta)P|h|^2}{\beta P|h|^2+\sigma^2}\right). \label{eq:information-lower-bound2}
\end{align}
This is an achievable rate measured in bits per channel use.

\section{MIMO Analysis and Capacity Lower Bound}
\label{section:MIMO}

Building on the approach developed for the SISO case, we now extend the analysis to a point-to-point MIMO channel with a deterministic channel matrix $\vect{H}\in \mathbb{C}^{M\times K}$, where $K$ and $M$ denote the numbers of transmit and receive antennas, respectively. After the receiver processing, the noisy data streams are quantized symbol by symbol before being forwarded to the baseband unit via a fronthaul link. We assume a total fronthaul capacity of $b_{\rm tot}$ bits per channel use. The central question is how to optimally select the number of spatially multiplexed data streams and the quantization resolution for each stream, based on the channel matrix and transmit power. To address this, we derive a lower bound on the MIMO capacity tailored to the assumed fronthaul quantization scheme using the mutual information lower bound derived in Section~\ref{section:SISO}.

We express the compact SVD of the channel matrix as $\vect{H} = \vect{U} \vect{\Sigma} \vect{V}^{\Htran}$, where $\vect{U} \in \mathbb{C}^{M \times r}$ and $\vect{V} \in \mathbb{C}^{K \times r}$ are the semi-unitary matrices containing the left and right singular vectors as their columns, respectively. These $r$ columns correspond to the $r$ non-zero singular values of $\vect{H}$, where $r\leq \min(M,K)$ is the rank of $\vect{H}$. The diagonal matrix $\vect{\Sigma} \in \mathbb{C}^{r \times r}$ contains the $r$ non-zero singular values of $\vect{H}$, arranged in descending order.

In the absence of fronthaul quantization noise, and assuming an additive white Gaussian noise (AWGN) channel, the received signal $\overline{\vect{y}}  \in \mathbb{C}^M$ is modeled as
\begin{align}
    \overline{\vect{y}} = \vect{H} \overline{\vect{x}} + \overline{\vect{n}},
\end{align}
where $\overline{\vect{n}} \sim \CN(\vect{0}, \sigma^2 \vect{I}_M)$ is the receiver noise.

Under a total transmit power constraint $\mathbb{E}\{\|\overline{\vect{x}} \|^2\} = P$, it is well known that the optimal transmission strategy is to send $\overline{\vect{x}} = \vect{V} \vect{x}$, where $\vect{x}$ contains independent entries with power values determined via the water-filling algorithm \cite{bjornson2024introduction}. At the receiver, the capacity-achieving processing is to apply $\vect{U}^{\Htran}$ to $\overline{\vect{y}}$, yielding $r$ independent parallel SISO channels that can be decoded separately. Applying the same transmit and receive processing and assuming Gaussian data signals, we can obtain a lower bound on the MIMO capacity under fronthaul quantization. We first write the $r$ parallel SISO channels as
\begin{align} \label{eq:MIMO-channel}
y_i = s_ix_i + n_i, \quad i=1,\ldots,r,
\end{align}
where $y_i$, $x_i$, $n_i$ are the $i$th elements of  $\vect{U}^{\Htran}\overline{\vect{y}}$, $\vect{x}$, and $\vect{U}^{\Htran}\overline{\vect{n}}$, respectively. Note that $n_1,\ldots,n_r \sim \CN(0,\sigma^2)$ are mutually independent since $\vect{U}$ is  unitary. The $r$ non-zero singular values of $\vect{H}$ are denoted in descending order, i.e., $s_1\geq s_2 \geq \ldots \geq s_r$. Since $\{x_i\}$ and $\{n_i\}$ are independent across $i$, the resulting samples 
$\{y_i = s_ix_i + n_i\}$ are independent across the index $i$. 
As a consequence, the quantized samples $\{z_i = Q(y_i)\}$ are also independent across $i$. By applying Lemma \ref{lemma:SISO}, which establishes that the distortion noise is uncorrelated with its corresponding data signal, we can obtain a lower bound on the capacity of the MIMO channel under fronthaul quantization by summing the achievable data rate from \eqref{eq:information-lower-bound2} for each stream. We obtain the following result.

\begin{corollary} \label{corollary}
Consider the MIMO channel in \eqref{eq:MIMO-channel}, where $x_i \sim\CN(0,p_i)$, $i=1,\ldots,r$. Suppose each received stream $y_i$ is quantized independently using a Lloyd-Max quantizer $Q_i(\cdot)$ with $b_i$ bits, which introduces a distortion factor $\beta_i$. Then, an achievable data rate for this quantized MIMO channel is
\begin{align}
R\left(\left\{p_i,b_i\right\}\right) =
\sum_{i=1}^{r} \log_2 \left(1 + \frac{(1 - \beta_i) p_i s_i^2}{\beta_ip_i s_i^2 + \sigma^2 } \right).
\end{align}
\end{corollary} 

This expression depends on the allocation of transmit powers and quantization bits over the data streams.
In the next section, we optimize these design variables.

\section{ Stream-Adaptive Joint Power and Bit Allocation}
\label{section:allocation}

The rate maximization problem, subject to a total quantization budget of $b_{\rm tot}$ bits per channel use (i.e., $\sum_{i=1}^r b_i = b_{\rm tot}$) and a total transmit power constraint $\sum_{i=1}^r p_i = P$, is formulated as follows:
\begin{align} \label{eq:optimization_problem}
\max_{\substack{\sum_{i=1}^r b_i = b_{\rm tot}, \quad b_i\geq 0, \forall i \\ \sum_{i=1}^r p_i = P, \quad p_i\geq 0, \forall i}} 
\sum_{i=1}^{r} \log_2 \left(1 + \frac{(1 - \beta_i) p_i s_i^2}{\beta_ip_i s_i^2 + \sigma^2} \right).
\end{align}

To design an efficient algorithm to solve \eqref{eq:optimization_problem}, we initially assume that $b_i$ can take any non-negative real value. 
We also notice that for realistically large $b_i$ (i.e., $b_i\geq 5$), the distortion factor $\beta_i$ for a Gaussian input to a Lloyd–Max quantizer can be well approximated using high-rate distortion theory \cite{quantization} as\footnote{Although we assume a Lloyd–Max quantizer, the $2^{-2b_i}$ distortion scaling law is equally valid for other high-rate scalar quantizers, differing only in the associated multiplicative constant, which depends on the source distribution and quantizer design \cite[Eqn. (8)]{quantization}.}
\begin{align}
    \beta_i \approx \frac{\sqrt{3}\pi}{2} \cdot 2^{-2b_i}. \label{eq:high-rate-distortion}
\end{align}
Using this expression, the approximate sum-rate maximization problem in terms of the newly defined optimization variables $\overline{b}_i^2=b_i$ and $\overline{p}_i^2=p_i$ (which eliminates non-negativity constraints) can be expressed as:
\begin{align} \label{eq:approx-problem}
    \max_{\substack{\sum_{i=1}^{r} \overline{b}_i^2 = b_{\rm tot} \\ \sum_{i=1}^{r} \overline{p}_i^2 = P}} 
\sum_{i=1}^{r} \log_2 \left(1 + \frac{\left(1 - \frac{\sqrt{3}\pi}{2}2^{-2{\overline{b}_i^2}}\right) \overline{p}_i^2 s_i^2}{\frac{\sqrt{3}\pi}{2}2^{-2{\overline{b}_i^2}}{\overline{p}_i^2} s_i^2 + \sigma^2} \right).
\end{align}
Taking the derivative of the Lagrangian function with respect to $\overline{b}_i$ yields that either $\overline{b}_i=0$ or otherwise the following optimality condition:
\begin{align}
    \frac{2^{2\overline{b}_i^2}}{\overline{p}_i^2s_i^2} = \lambda,
\end{align}
where $\lambda > 0$ is a constant ensuring that $\sum_{i=1}^r \overline{b}_i^2 = b_{\rm tot}$. Taking the base-2 logarithm of both sides, we obtain for any non-zero $\overline{b}_i^2$
\begin{align}
   b_i=\overline{b}_i^2 = \underbrace{\log_2\left(\sqrt{\lambda}\right)}_{\triangleq \mu} + \log_2\left(\sqrt{p_i s_i^2}\right),
\end{align}
where $\mu$ is chosen to satisfy the total bit constraint, leading to the value $\mu^{\star}$. Thus, the optimal real-valued bit allocations are
\begin{align}
    b_i^{\star} = \max\left(0, \mu^{\star} + \log_2\left(\sqrt{p_i s_i^2}\right)\right), \quad i=1,\ldots,r.
\end{align}
Since $b_i$ must be an integer in practice, we can round $b_i^{\star}$ to the nearest integer to obtain a judicious near-optimal bit allocation.

To gain insights into the structure of the optimal power allocation, assume that $b_i^{\star} > 0$ for $i = 1, \ldots, r'$. Plugging $b_i^{\star}$ into the rate expression, we get
\begin{align}
    2^{-2b_i^{\star}} p_i s_i^2 = \frac{1}{\lambda^{\star}}, \quad \text{where } \lambda^{\star} = 2^{2\mu^{\star}}.
\end{align}
This leads to the following sum-rate expression in terms of $\overline{p}_i^2=p_i$
\begin{align}
    \max_{\sum_{i=1}^{r'} \overline{p}_i^2 = P} 
\sum_{i=1}^{r'} \log_2 \left(\frac{\overline{p}_i^2 s_i^2 + \sigma^2}{\frac{\sqrt{3}\pi}{2\lambda^{\star}}  + \sigma^2} \right).
\end{align}
Differentiating the Lagrangian corresponding to the above problem with respect to $\overline{p}_i$ yields $\overline{p}_i=0$ or otherwise for non-zero $\overline{p}_i$
\begin{align}
    p_i + \frac{\sigma^2}{s_i^2} = \lambda', \label{eq:pi_opt}
\end{align}
where $\lambda'$ is a constant. Then, the power coefficients are obtained as
\begin{align}
p_i = \max\left(0,\lambda'- \frac{\sigma^2}{s_i^2}\right), \quad i=1,\ldots,r',
\end{align}
where $\lambda'$ ensures $\sum_{i=1}^{r'} p_i = P$.

Based on the above analysis, an algorithmic framework can be established as follows. For each value of \( r' \in \{ 1, \ldots, r \} \), compute the power coefficients \( p_i \) as described above. Next, we determine the corresponding \( b_i^{\star} \), round them to the nearest integers, adjust for any surplus or deficit in the total bit count, and evaluate the resulting sum rate. The configuration that yields the highest sum rate is declared as the near-optimal solution. The complete procedure is summarized in Algorithm~1.

\begin{algorithm}
\caption{Bit and Power Allocation via Joint Optimization }
\begin{algorithmic}[1]
\State \textbf{Input:} Total transmit power $P$, total number of bits $b_{\rm tot}$, noise variance $\sigma^2$, channel strengths $s_i$ for $i=1,\ldots,r$
\State \textbf{Initialize:} $\text{best\_rate} \gets 0$, $r_{\max} \gets r$
\For{$r' = 1$ to $r_{\max}$}
    \State Compute $p_i = \max\left(0, \lambda' - \frac{\sigma^2}{s_i^2}\right)$ for $i=1,\ldots,r'$, where $\lambda'$ is selected to have $\sum_{i=1}^{r'} p_i = P$
    \State Denote the number of active subchannels for which $p_i>0$ by $r''$
    \State Compute $\mu^{\star}$ such that $\sum_{i=1}^{r''} b_i^{\star} = b_{\rm tot}$ with
    \[
    b_i^{\star} = \max\left(0, \mu^{\star} + \log_2\left(\sqrt{p_i s_i^2}\right)\right)
    \]
    \State Round $b_i^{\star}$ to nearest integer to obtain $\tilde{b}_i$
    \State Compute $\Delta = \sum_{i=1}^{r''} \tilde{b}_i - b_{\rm tot}$

    \If{$\Delta > 0$}  \Comment{Too many bits assigned}
        \State \textbf{Greedy Bit Removal:}
At each step, a bit is removed from the stream that causes the smallest rate loss, and this is repeated until the total bit budget is satisfied.
    \ElsIf{$\Delta < 0$}  \Comment{Not enough bits assigned}
        \State \textbf{Greedy Bit Addition:}
        At each step, a bit is added to the stream that causes the largest rate improvement, and this is repeated until the total bit budget is satisfied.
    \EndIf

    \State Compute achievable sum rate with final $\tilde{b}_i$ and $p_i$
    \If{sum rate $>$ best\_rate}
        \State Store current configuration as best
        \State Update $\text{best\_rate}$
    \EndIf
\EndFor
\State \Return $\{p_i, b_i\}$ yielding highest sum rate
\end{algorithmic}
\end{algorithm}

\section{Asymptotic Behavior at High SNRs} \label{eq:section-asymptotic}

In the high-SNR regime where \( s_i^2 / \sigma^2 \to \infty \) for \( i = 1, \ldots, r \), the approximate sum rate in \eqref{eq:approx-problem} approaches the upper limit
\begin{align}
    \sum_{i=1}^{r} \log_2 \left(1 + \frac{1 - \frac{\sqrt{3}\pi}{2}2^{-2b_i}}{\frac{\sqrt{3}\pi}{2}2^{-2b_i}} \right), \label{eq:sumrate-highSNR}
\end{align}
which is independent of the power allocation. By differentiating the Lagrangian corresponding to \eqref{eq:sumrate-highSNR} under the total bit constraint \( \sum_{i=1}^{r} b_i = b_{\rm tot} \), it follows that the optimal solution in the high SNR regime is to assign an equal number of bits to each active stream: $b_i = b_{\rm tot}/r$. 

The question is how many streams, $r'$, to activate at a high but finite SNR. We will propose an algorithm that determines the solution by activating one stream at a time. Assume that in the $r'$th iteration, the strongest $r'$ data streams are active. When \( b_{\rm tot} \) is not divisible by \( r' \), perfect uniformity may not be feasible. In this case, the remaining bits must be distributed across the active streams, and we assign them to the strongest subchannels (i.e., those with the largest \( s_i^2 \)) in analogy with water-filling, which allocates more power to stronger subchannels. Hence, each active stream is quantized using either \( b^{\star} \) or \( b^{\star} + 1 \) bits, where \( b^{\star} = \left\lfloor \frac{b_{\rm tot}}{r} \right\rfloor \) denotes the largest integer less than or equal to \( \frac{b_{\rm tot}}{r} \). The bit assignments are chosen such that the total number of quantization bits remains equal to \( b_{\rm tot} \). 

Under a uniform quantization resolution with a fixed number of bits per stream, which is asymptotically optimal, the data rate maximization problem reduces to a power allocation problem:
\begin{align}
    \max_{\sum_{i=1}^{r'} p_i = P,\quad p_i\geq 0, \forall i} 
    \sum_{i=1}^{r'} \log_2 \left(1 + \frac{(1 - \beta_i) p_i s_i^2}{\beta_ip_i s_i^2 + \sigma^2 } \right). \label{eq:power-allocation}
\end{align}
Defining $g_i(x) = \log_2\left(1+\frac{(1-\beta_i)x}{\beta_i x+1}\right)$, we can write \eqref{eq:power-allocation} as
\begin{align}
    \max_{\sum_{i=1}^{r'} p_i = P,\quad p_i\geq 0, \forall i} 
    \sum_{i=1}^{r'} g_i\left(p_i\frac{s_i^2}{\sigma^2}\right).
\end{align}
Noting that $g_i(x)$ is a concave function and has an invertible derivative function for non-negative $p_i\geq 0$, we have the optimal power allocation from \cite[Thm.~3.16]{Bjornson2013d} as
\begin{align}
    p_i=\max\left(0, \frac{\sigma^2}{s_i^2}g_i'^{-1}\left( \frac{\sigma^2}{\nu s_i^2}\right)\right), \quad i=1,\ldots, r',
\end{align}
where $\nu\geq 0$ is selected to have $\sum_{i=1}^{r'} p_i=P$. The optimal value of $\nu$ can be found by a bisection search by noting that $g_i'^{-1}(\nu)$ is an increasing function of $\nu$.

We note that if some subchannels remain inactive after the power allocation step, their unused bit budget is redistributed uniformly among the active subchannels. This procedure is repeated for $r'=1,\ldots,r$, where for each value of $r'$, only the $r'$ strongest subchannels are active. At the end of the algorithm, the candidate that results in the best sum rate is declared as the solution.

\section{Numerical Results}
\label{section:results}

In this section, we evaluate the performance of the proposed joint bit and power allocation algorithm by comparing it with several benchmark schemes. The result of Algorithm~1 is denoted by ``JBP-Alloc,'' referring to joint bit and power allocation. As a reference, we also consider the ideal case without fronthaul quantization, where the optimal solution is obtained using the classical water-filling algorithm. This upper-bound performance is denoted as ``Ideal'' in the figures. As another benchmark, we include the quantization-unaware water-filling solution, where the bits are uniformly distributed across the active data streams, determined by the power allocation as if there were no quantization, following the classical water-filling rule. However, the achieved rate is still evaluated using the actual quantized model. In the figures, this method is labeled as ``UnawareWF''.

In addition, we include a uniform bit allocation scheme motivated by the asymptotic high-SNR analysis in Section~\ref{eq:section-asymptotic}. 
This algorithmic implementation, referred to as ``UB-Alloc'', enforces uniform integer-valued bit allocation and evaluates the achievable rate using the exact finite-SNR expression. Lastly, we consider a higher-complexity greedy bit allocation algorithm as a benchmark. In this approach, bits are assigned one by one to the subchannel that yields the highest marginal improvement in the sum rate. After each bit assignment, optimal power allocation is computed using the same method as in Section~\ref{eq:section-asymptotic}.  At each of the $b_{\mathrm{tot}}$ steps, the algorithm tests all $r$ candidate streams by temporarily adding one bit and recomputing the resulting optimal power allocation via a bisection-based routine.
Hence, the bisection-based power allocation is executed $b_{\mathrm{tot}}\cdot r$ times in total. On the other hand, the proposed JBP-Alloc algorithm requires implementing bisection search only $2r$ times. This greedy scheme is denoted by ``Greedy'' in the figures.

For \( b = 1, \ldots, 5 \), we use the distortion factors of Lloyd–Max quantizers provided in \cite[Tab.~4.1]{jain1989fundamentals}. For higher bit values, the distortion factor \( \beta \) is computed using the high-rate distortion approximation given in~\eqref{eq:high-rate-distortion}. We consider a MIMO channel with $K=16$ transmit antennas and $M=128$ receive antennas, both arranged as uniform linear arrays.

We define the SNR per antenna as $\mathrm{SNR} = \frac{P \sum_{i=1}^{r} s_i^2}{MK\sigma^2}$.
In the figures, the channel matrices 
$\vect{H}$ follow a Rician fading model with 200 NLOS paths and one LOS path, where $\kappa$ denotes the Rician $\kappa$-factor. The complex path gains of the NLOS components are modeled as independent, zero-mean complex Gaussian variables with equal variance across the paths. The singular values are subsequently scaled to satisfy the specified SNR, i.e., $
\mathrm{SNR} = \frac{P \sum_{i=1}^{r} s_i^2}{MK\sigma^2}=10$\,dB. 
All results are averaged over 100 independent channel realizations.

In Fig.~\ref{fig1}, we set $\kappa = 0$\,dB and plot the sum rate as a function of $b_{\rm tot}$. Since the ``Ideal'' case is independent of the number of quantization bits, its performance appears as a horizontal line and, as expected, serves as an upper bound to the achievable data rate under quantization. For $\kappa = 0$\,dB, the singular values are relatively similar, which explains why the UB-Alloc method, allocating bits uniformly, achieves almost the same performance as the JBP-Alloc method. Furthermore, the ``Greedy'' method yields the same performance but at a significantly higher computational cost compared to the proposed methods. In contrast, neglecting quantization and applying classical water-filling results in up to $20\%$ loss in sum rate performance. Moreover, it is observed that, under a limited bit budget, the quantization-unaware scheme activates more subchannels than the proposed scheme.

\begin{figure}[t!]
        \centering
	\begin{overpic}[width=0.99\columnwidth,trim=0cm 0cm 0.5cm 0.5cm,clip,tics=10]{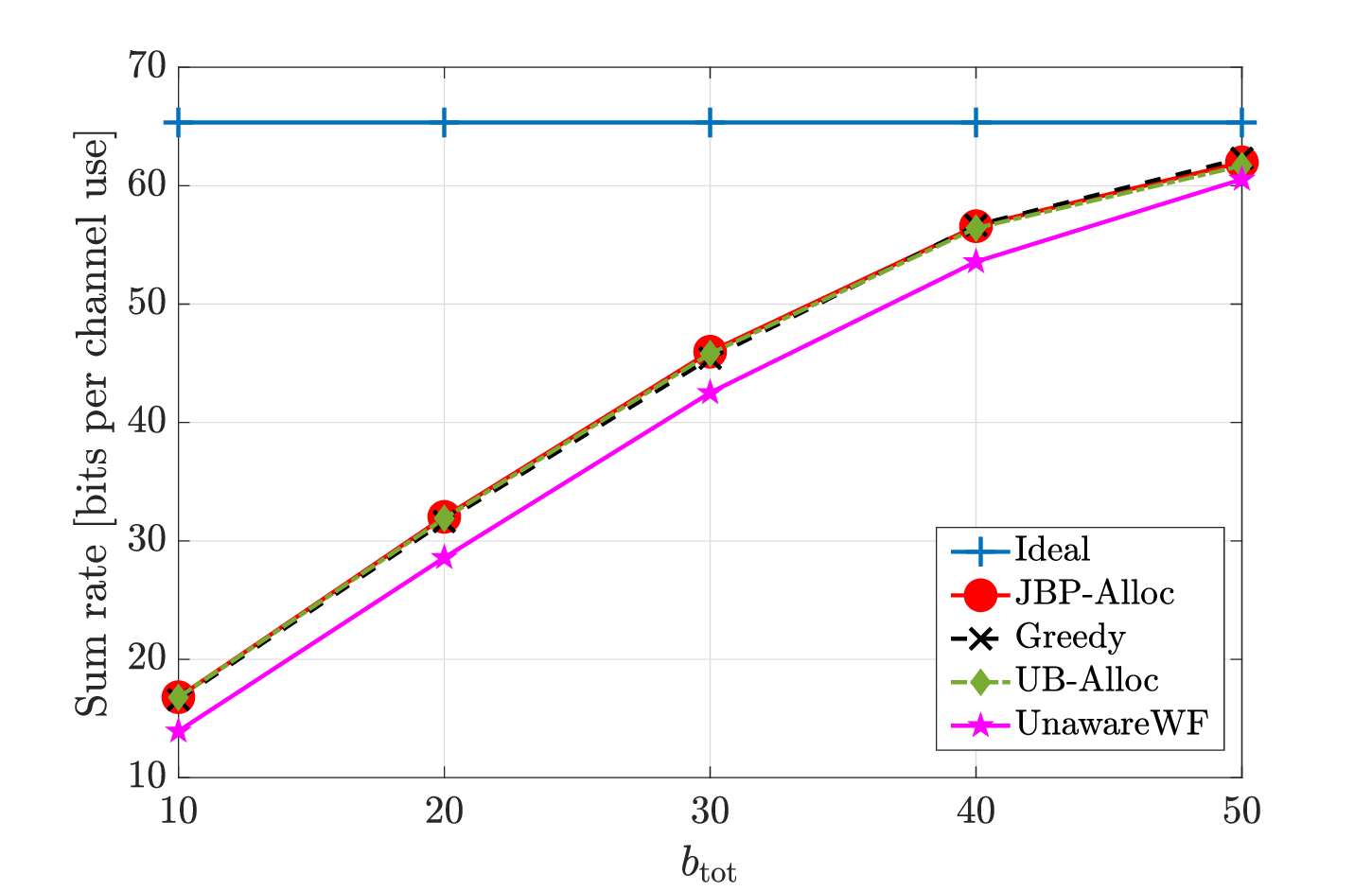}
\end{overpic} 
\vspace{-2mm}
        \caption{The sum rate in terms of the total number of quantization bits for $\kappa=0$\,dB. }
        \label{fig1}
        \vspace{-4mm}
\end{figure}

\begin{figure}[t!]
        \centering
	\begin{overpic}[width=0.99\columnwidth,trim=0cm 0cm 0.5cm 0.5cm,clip,tics=10]{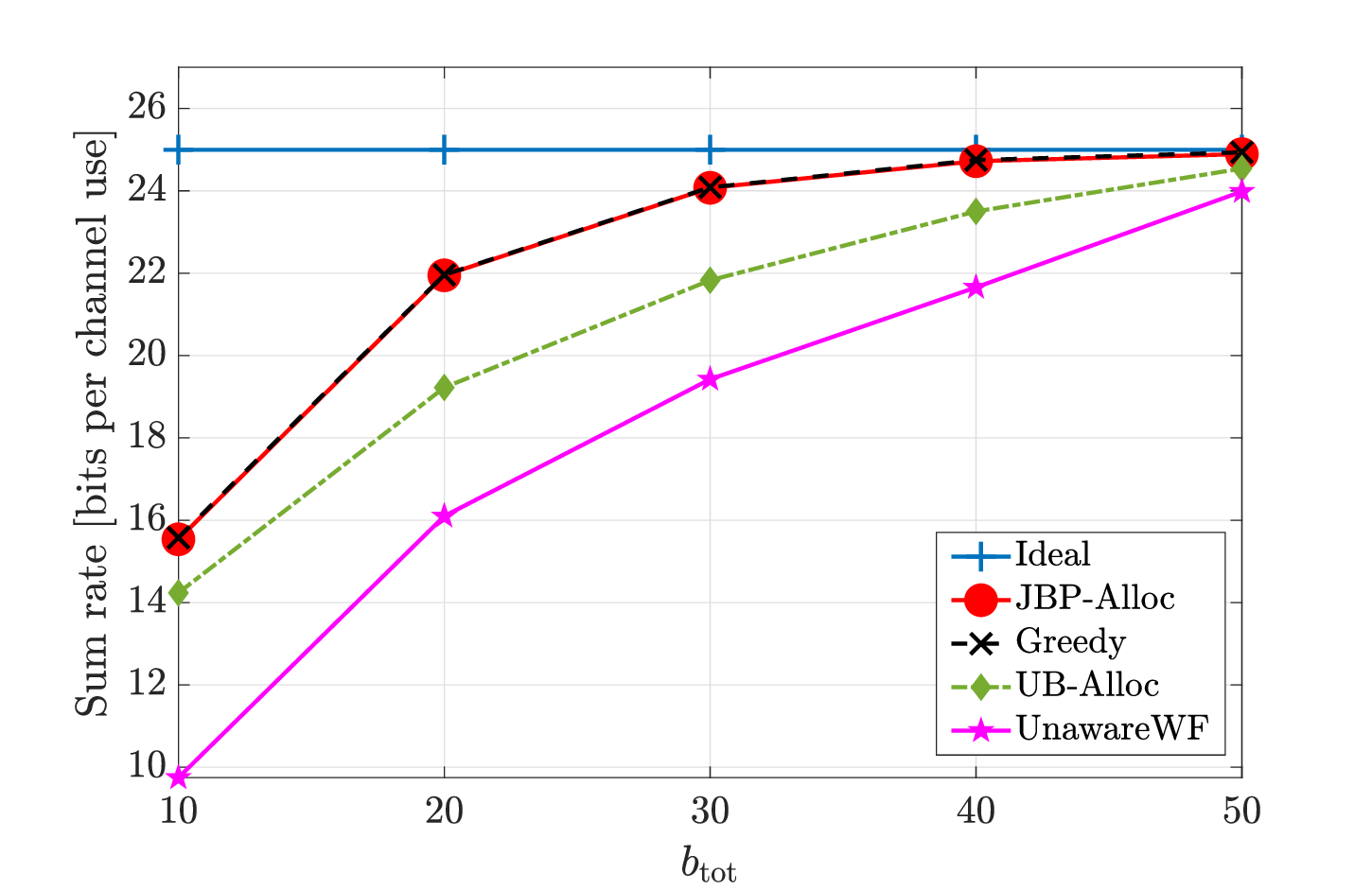}
\end{overpic} 
\vspace{-2mm}
        \caption{The sum rate in terms of the total number of quantization bits for $\kappa=20$\,dB.  }
        \label{fig2}
        \vspace{-4mm}
\end{figure}

In Fig.~\ref{fig2}, we increase $\kappa$ to $20$\,dB, which leads to larger variations among the singular values of the channel. As a result, the sum rate decreases since fewer subchannels are useful for transmission. In this scenario, a much smaller number of subchannels are expected to remain active, and thus 100 bits are sufficient for the proposed JBP-Alloc method to achieve the performance of the ``Ideal'' case. The ``Greedy'' method provides almost the same performance but at a considerably higher computational complexity. Unlike the previous figure, we now observe a performance degradation with UB-Alloc, since allocating bits uniformly is no longer optimal. More importantly, the UnawareWF method yields a significantly smaller data rate. These results highlight the importance of the proposed JBP-Alloc method, which consistently achieves the best performance across all considered scenarios while requiring $b_{\rm tot}/2$ times less computational complexity than the ``Greedy'' technique.

\section{Conclusions}
\label{section:conclusions}

We investigated stream-adaptive quantization for fronthaul-constrained MIMO systems, where the streams are decoupled in the radio unit, but signal decoding is done at the baseband unit.  The proposed JBP-Alloc algorithm was shown to efficiently balance bit and power allocation, achieving nearly optimal performance with substantially lower computational complexity than greedy allocation. Our asymptotic analysis revealed that uniform bit allocation becomes optimal at high SNR, simplifying system design in this regime. Numerical results highlight the importance of stream-adaptive fronthaul quantization and demonstrate that significant performance gains can be achieved by carefully adapting quantization and power allocation to channel conditions.

\bibliographystyle{IEEEtran}
\bibliography{IEEEabrv,refs}

\end{document}